\documentclass[preprint,12pt,authoryear]{elsarticle}

\usepackage{amssymb, amsmath, amsthm}
\usepackage{array}
\usepackage{amsfonts}
\usepackage{eucal}
\usepackage{nccmath}
\usepackage{bbm}
\usepackage{natbib}
\bibpunct{[}{]}{;}{n}{,}{,}

\newtheorem{theorem}{Theorem}
\newtheorem{lemma}{Lemma}

\newtheorem{definition}[lemma]{Definition}

\newcommand{\wt}[1]{\widetilde{#1}}
\newcommand{\trinom}[4]{\medmath{\binom{#1}{\,#2,\; #3,\; #4\,}}}
\newcommand{\Q}{1-\ell\left( \sqrt{r} h +  s +  \sqrt{r} h^{-1}\right)}
\newcommand{\mpmon}{h + s + r h^{-1}}
\newcommand{\mpsym}{\sqrt{r} h + s + \sqrt{r} h^{-1}}

\renewcommand{\d}[0]{\mathrm{d}}
\newcommand{\norm}[1]{\left\| #1 \right\|}
\newcommand{\matr}[4]{\left( \begin{array}{ll} #1 & #2 \\ #3 & #4 \end{array}\right)}
\newcommand{\ip}[1]{\left\langle #1 \right\rangle}

\newcommand{\vect}[2]{\left(\begin{array}{c} #1 \\ #2 \end{array}\right)}

\begin{document}

\begin{frontmatter}


\author{Patrick Waters\corref{cor1}\fnref{label2}}
\ead{pwaters@math.arizona.edu}
\fntext[label2]{University of Arizona}
\cortext[cor1]{Corresponding author}

\title{Solution of string equations for asymmetric potentials}



\begin{abstract}
We consider the large $N$ expansion of the partition function for the Hermitian one-matrix model.  It is well known that the coefficients of this expansion are generating functions $F^{(g)}$ for a certain kind of graph embedded in a Riemann surface.  Other authors have made a simplifying assumption that the potential $V$ is an even function.  We present a method for computing $F^{(g)}$ in the case that $V$ is not an even function.  Our method is based on the string equations, and yields ``valence independent'' formulas which do not depend explicitly on the potential.  We introduce a family of differential operators, the ``string polynomials'', which make clear the valence independent nature of the string equations.
\end{abstract}


\end{frontmatter}


\section{Introduction} \label{Intro}
We consider a partition function defined by 
\begin{align}
Z_{n,N}(t)=&  \int_{\mathbb{R}^n}\exp\left[-N\sum_{i=1}^n  V(\lambda_i ;t)\right]\prod_{1\leq i<j\leq n} (\lambda_i-\lambda_j)^2 \, \d^n \lambda . \label{ZDef1}
\end{align}
We assume that the potential function $V$ is a polynomial:
\begin{align}
V(\lambda;t) =& \frac{1}{2} \lambda^2 + \sum_{j=1}^{d} t_j \lambda^j . \label{VDef1}
\end{align}
For the integral in (\ref{ZDef1}) to converge, it is necessary that the polynomial degree $d$ is even, and that the parameter $t_d$ is such that the real part of $t_{d}$ is positive.  By replacing the domain of integration $\mathbb{R}^n$ with $\Gamma^n$, where $\Gamma$ is an appropriately chosen contour in the complex plane, it is possible to remove these restrictions \cite{Bleher10}.

The partition function in (\ref{ZDef1}) is the normalization constant for the induced measure on eigenvalues arising from the following probability measure on $n\times n$ Hermitian matrices:
\begin{align}
dP_{n}(M) =& \frac{1}{\wt{Z}_{n,N}(t)} \exp \left[-N V(M;t)\right]\, \d M. \label{MatrixMeasure1}
\end{align}
The factor $\d M$ is Lebesgue product measure $n^2$ real degrees of freedom of $M$. 

We are interested in the partition function (\ref{ZDef1}) because it has a topological expansion whose coefficients are generating functions for a certain kind of embedded graph call a \emph{map}.  Letting $n,N \rightarrow \infty$ with the ratio $n/N=x$ held fixed, we have
\begin{align}
\log \frac{Z_{n,N}(t)}{Z_{n,N}(0)} \sim & \sum_{g \geq 0} F^{(h)}(x,t)N^{2-2g}, \label{TopExp1} \\
F^{(g)} =& \sum_{\Gamma \in \substack{  \text{genus $g$ maps}  \\  \text{vert's of valence $\leq d$}   }} x^{\text{Faces}(\Gamma)}\prod_{j=1}^{d} \frac{(-t_j)^{V_j(\Gamma)}}{\scriptstyle{V_j(\Gamma)!}}. \label{GenFunc1}
\end{align}
Here $V_j(\Gamma)$ is the number of $j$-valent vertices of $\Gamma$; that is, vertices at which exactly $j$ edges meet.
We use $\sim$ to mean that for any $g\geq 0$, truncating the series after terms of order $N^{2-2g}$ results in an approximation with error $O(N^{1-2g})$ as $n,N\rightarrow\infty$ with $n/N=x$.

A \emph{map of genus $g$} is an equivalence class of labeled graphs embedded in genus $g$ compact connected oriented surface.  
Two embedded graphs are equivalent if an orientation preserving homeomorphism of the surface takes the vertices and edges of one to those of the other, preserving labels.  A map is equipped with the following labels: each of the map's $k_j$ vertices of each valence $j$ has a unique label $1,\ldots ,k_j$; and also for each vertex, one incident edge is marked with an arrow pointing away from that vertex.  This marking of edges can be represented as a function $f$ from the vertex set to the edge set such that each vertex $V$ is incident $f(V)$.

The expansion (\ref{TopExp1}) was used by Bessis, Itzykson and Zuber \cite{BIZ}.  Ercolani and McLaughlin \cite{EMcL03} gave a rigorous proof that an expansion of the form (\ref{TopExp1}) exists, and proved analytical properties of the expansion.  In particular, expansions of derivatives of $\log Z_{n,N}(t)/Z_{n,N}(0)$ with respect to $g$ can be computed by termwise differentiation of the series (\ref{TopExp1}).  The same is true for differentiation with respect to $x$, but for more complicated reasons; see section 4.1.3 of \cite{PatricksDissertation} for a proof that this property is equivalent to the string equations.

This paper addresses the problem of computing the generating functions $F^{(g)}$, a problem which has been studied by several authors \cite{BIZ,BeyondSphericalLimit,Shirokura95,Eynard04,EMcLP08,Ercolani11,EP11,Bleher10,AAM11,PatricksDissertation}.   These articles give a variety of methods for computing $F^{(g)}$, and yield formulas given in terms of a variety of choices of variables.  In this article we focus on \emph{valence independent} formulas.  Formulas of this kind allow derivatives with respect to $x$ but forbid explicit dependence on the potential $V$ or its parameters $t_j$.  For example we find that
\begin{align}
F^{(1)} =& -\frac{1}{12} \log \frac{z}{x} +\frac{1}{24} \log \left( (u')^2 z -(z')^2 \right), \label{F1Formula}
\end{align}
where $z$ and $u$ are defined in equations (\ref{rnExpand}-\ref{snExpand}) and the notation $'$ indicates $\partial_x$.  Other authors assume that the potential $V$ is an even function; this assumption significantly reduces the difficulty of computing $F^{(g)}$, for example with this assumption we have $u=0$.  Thus our formula (\ref{F1Formula}) is new and extends formulas given by other authors to the case of an asymmetric potential.  The valence independent approach was first used by Shirokura \cite{Shirokura95}.

For the purpose of computing the generating functions $F^{(g)}$, it is difficult to work with the partition function directly.  Szego's relation (\ref{SzegoRel1}) and the Hirota-Szego relation (\ref{SzegoRel2}) give two fundamentally different relations between the partition function and recurrence coefficients for a family of orthogonal polynomials:
\begin{align}
\frac{Z_{n+1,N} Z_{n-1}}{Z_{n,N}^2} =& \frac{n+1}{n} r_{n,N}. \label{SzegoRel1} \\
r_{n,N} \sim & \frac{1}{N^2}\frac{\partial^2}{\partial t_{1}^2} \log \frac{Z_{n,N}(t)}{Z_{n,N}(0)}   \label{SzegoRel2}
\end{align}
Strategies for computing $F^{(g)}$ may involve using string equations to compute higher order asymptotics of the recurrence coefficients $r_{n,N}$, then lifting this information to the level of the partition function $Z_{n,N}$.  For example \cite{BIZ} uses (\ref{SzegoRel1}), \cite{EMcLP08} uses (\ref{SzegoRel2}) and \cite{AAM11} uses Bleher-Its deformation.  We will use (\ref{SzegoRel1}); the mathematical rigor of our use of this formula was established in section 4.1.4 of \cite{PatricksDissertation}.  Strategies for computing $F^{(g)}$ which are based on loop equations \cite{BeyondSphericalLimit,Eynard04} avoid the recurrence coefficients $r_n$, but involve computing higher order asymptotics of random matrix correlators.  It is also possible to compute $F^{(g)}$ directly from the Riemann-Hilbert problem for orthogonal polynomials \cite{PatricksDissertation}.

Let us now define the coefficients $r_{n}$ and explain the string equations which govern them.  The monic orthogonal polynomials $p_{n,N}$ arise from the following (non-Hermitean) inner product and orthogonality relation:
\begin{align}
\ip{F(\lambda),G(\lambda)}=& \int_{\mathbb{R}} F(\lambda) G(\lambda)\exp\left(-N V(\lambda;t)\right)\, d\lambda . \label{OrthRel1} \\
\ip{p_{n,N}(\lambda),\lambda^m}=&0 &\text{ for all }m<n. \label{OrthRel2}
\end{align}
Thus the polynomials $p_n$ can be constructed by the Gram-Schmidt process.  As is well known \cite{DeiftOPs}, the polynomials satisfy a three term recurrence:
\begin{align}
\lambda p_{n,N}(\lambda) =& p_{n+1,N}(\lambda) + s_{n,N} p_{n,N}(\lambda) + r_{n,N} p_{n-1,N}(\lambda). \label{ThreeTermRec1}
\end{align}
This defines the recurrence coefficients $r_{n}$.
The three term recurrence (\ref{ThreeTermRec1}) can be expressed in terms of a Jacobi matrix $L$:
\begin{align} \label{jacobiMatrix}
L=\left( \begin{array}{ccccc}
s_0 & 1 & 0 & 0 &  \\
r_1 & s_1 & 1 & 0 &  \\
0 & r_2 & s_2 & 1 & \\
 & & \ddots & \ddots & \ddots
 \end{array}\right)
 \end{align}
Let $p= ( p_n(\lambda))_{n\geq 0}$ be the column vector of all orthogonal polynomials for the potential $V$; then the three term recurrence can be encoded as a matrix equation
\begin{align}
\lambda p =& L p. \label{matrixRecurrence}
\end{align}
Following Blecher and Dea\~{n}o \cite{Bleher10}, we use the string equations
\begin{align}		
0 =&  V'(L)_{n,n},	\label{UnDiff1}		\\
\frac{n}{N} =&  V'(L)_{n,n-1}.	\label{UnDiff2}		
\end{align} 
A differenced form of these equations was used in \cite{EMcLP08}.
Equation (\ref{UnDiff1}) can be derived as follows.  Integrating by parts we find that
\begin{align*}		
0 =& \int p'_n(\lambda) p_n(\lambda)\, e^{-NV(\lambda)}\,d\lambda	\\
=& N \int p_n(\lambda)^2  V'(\lambda)\, e^{-NV(\lambda)}\,d\lambda 	\\
=& N \norm{p_n}^2 V'(L)_{n,n}.
\end{align*}
An analogous calculation, but starting with $n \norm{p_{n-1}}^2 = \int p'_{n}(\lambda) p_{n-1}(\lambda)e^{-NV(\lambda)} \, d\lambda$ establishes (\ref{UnDiff2}).

We now explain how the continuum limit of string equations can be used to compute the generating functions $F^{(g)}$.  The following asymptotic expansions for recurrence coefficients with shifted indices was proven in \cite{EMcLP08}:
\begin{align}
r_{n+k,N} \sim & \exp \left(k N^{-1}\partial_x\right) \sum_{g \geq 0} z_g(x,t)N^{-2g} \label{rnExpand}, \\
s_{n+k,N} \sim & \exp \left(k N^{-1}\partial_x\right) \sum_{g \geq 0} u_g (x,t)N^{-g} \label{snExpand}.
\end{align}
Thus we have defined $z_g$ and $u_g$ to be the asymptotic coefficients in the expansions (\ref{rnExpand}-\ref{snExpand}).  We will use the abbreviations $z=z_0$ and $u=u_0$.  Inserting these expansion in the string equations gives the continuum limit of the string equations.  By extracting terms at order $N^{-2g}$ from the continuum limit of string equations, one can solve for $z_g$ and $u_{2g}$ in lower order terms.  Inserting the expansions (\ref{rnExpand}-\ref{snExpand}) and the topological expansion (\ref{TopExp1}) in the Szego relation (\ref{SzegoRel1}), one may solve for $F^{(g)}$ in terms of the variables $z_{g'},\, g'\leq g$:
\begin{align}
F^{(g)} =& \sum_{m=0}^{g} \frac{(1-2m)B_{2m}}{(2m)!} \partial_{x}^{2m-2} \wt{z}_{g-m} \label{AsymptSzegoRel1}
\end{align}
The constants $B_{2m}$ are Bernoulli numbers, and the cumulants $\wt{z}_g$ are defined by
\begin{align}
&\sum_{g \geq 0} \wt{z}_g N^{-2g}  \sim  \log \sum_{g \geq 0}x^{-1}z_g N^{-2g}. \\
&\begin{array}{ *4{>{\displaystyle}c |}  >{\displaystyle}c }
g & 0 & 1 & 2 & 3    \\ \hline
& & & & \\[-2.1ex]
\wt{z}_g & \log \frac{z}{x} & \frac{z_1}{z} & \frac{z_2}{z}-\frac{z_{1}^2}{2z^2} & \frac{z_3}{z}-\frac{z_2 z_1}{z^2}+\frac{z_{1}^3}{3z^3} 
\end{array}
\end{align}
In \cite{BIZ}, formula (\ref{AsymptSzegoRel1}) was obtained using the Euler-MacLaurin formula.  As was noted in \cite{AAM11}, further justification is required to make this approach rigorous.  A rigorous proof of (\ref{AsymptSzegoRel1}) has been given in section 4.2.1 of \cite{PatricksDissertation}.

The rest of our paper is organized as follows.  In section 2 we prove a theorem concerning the structure of the string equations; theorem \ref{str polyn thm} describes the terms arising in the string equations as a sum of differential operators (which we call string polynomials) acting on a residue.  The reason that our method results in valence independent formulas is essentially that the string polynomials do not depend on the potential.  In section 3 we demonstrate our method by computing $F^{(1)}$.  In section 4 we prove a structure theorem for  valence independent formulas for $u_g$ and $z_g$.  To compute $F^{(g)}$ using formula (\ref{AsymptSzegoRel1}) it is necessary to evaluate an indefinite integral with respect to $x$.  It turns out that in the cases $g=1,2$ this integral can be computed in closed form, but we have no proof that this happens for all $g>0$.

\section{String polynomials} \label{solving string}

%

Let us sketch an outline of this section.  In \ref{str polyn sect} we state theorem \ref{str polyn thm}, asserting that the $x$-derivative structure of the string equations can be expressed in terms of a family of valence independent ``string polynomials''.  We then reduce theorem \ref{str polyn thm} to lemma \ref{alt str polyn thm}, which essentially allows us to work with each time variable separately.  In sections \ref{Motzkin sect} - \ref{string general case sect} we prove lemma \ref{alt str polyn thm}.  Our proof gives a formula for the string polynomials, which is stated in lemma \ref{str polyn formula}.

\subsection{String polynomials} \label{str polyn sect}
In this section we will use the following notations
\begin{align}
r_{n,N} \sim& r=\sum_{h\geq 0} z_g N^{-2g}, \\
s_{n,N} \sim & s = \sum_{g \geq 0} u_g N^{-g}.
\end{align}
That is, we are defining $r$ and $s$ to be the formal series for the asymptotic expansions of $r_{n,N}$ and $s_{n,N}$.  Note that $\sim$ has the same meaning as in (\ref{GenFunc1}); let us give two examples marking the distinction between $\sim$ and $=$.  First $\exp(-N)\sim  \sum_{g\geq 0} 0 \times N^{-g}=0$, where the sum converges to the ``wrong'' limit; and second $\int_{0}^\infty (1+\tau/N)^{-1}\exp(-\tau)\, \d \tau \sim \sum_{g\geq 0} (-1)^g g! N^{-g}$, where the sum diverges (this example is the Stieltjes function, see \cite{BoydHyper}).  

We will use the notation $\sqrt{r}$ for the asymptotic expansion of $\sqrt{r_n}$, where we choose the branch of the square root that is positive when the potential is real.  The square root is natural because if one works with orthonormal polynomials $P_n$ (with positive leading coefficients) instead of monic orthogonal polynomials $p_n$, the three term recurrence becomes $\lambda P_n =\sqrt{r_{n+1}}P_{n+1} +s_n P_{n} +\sqrt{r_n}P_{n-1}$.  Here $r_n$ and $s_n$ are the same coefficients as in the case of monic polynomials.

The variables $r$ and $s$ allow us to analyze the $x$-derivative structure of the string equations before stratifying terms by genus.  We will use integer partitions as multi-indices to describe differentiation:
\begin{align}
\partial_{x}^{\lambda} f(x) := \prod_{i=1}^{\text{len}(\lambda)} \partial_{x}^{\lambda_i}f(x).
\end{align}
Furthermore, since $x$-derivatives in (\ref{rnExpand}-\ref{snExpand}) are always accompanied by powers of $N^{-1}$ it natural to differentiate with respect to $n=xN$, that is $\partial_{n}f(x):=N^{-1}f'(x)$.  Thus for example $\partial_{n}^{(1+2)}S=N^{-3}S'S''$.

We will use the notation $[h^p]$ for the $h^p$ coefficient of a Laurent series in powers of $h$.  That is,
\begin{align}
[h^p]G(h) =& \frac{1}{2\pi i} \oint_0 G(h)h^{-p} \frac{\d h}{h}. \nonumber
\end{align}

In section \ref{string general case sect}, we will prove the following theorem.
\begin{theorem}[Existence of string polynomials] \label{str polyn thm}
There are polynomials (with coefficients in the ring $\mathbb{Q}[r]$) $P^{(a)}_{\lambda,\eta}(\partial_r,\partial_s)$ indexed by integer partitions $\lambda,\eta$ such that the continuum string hierarchy can be expressed as
\begin{align}
0\sim & \sum_{\lambda,\eta}  \partial_{n}^{\lambda}s\,  \partial_{n}^{\eta} r \,  P^{(a)}_{\lambda,\eta}(\partial_{s}, \partial_{r}) \cdot [h^0] V'( h + s + r h^{-1}) \\
x \sim & \sum_{\lambda,\eta}  \partial_{n}^{\lambda}s\, \partial_{n}^{\eta} r\, P^{(b)}_{\lambda,\eta}(\partial_{s}, \partial_{r}) \cdot [h^0] V'( h + s + r h^{-1})
\end{align}
The string polynomials $P^{(a \text{ or }b)}_{\lambda,\eta}$ do not depend on the potential $V$.  Furthermore, they can be chosen to have polynomial degree at most one in $r$, and are uniquely determined by this condition.
\end{theorem}

We now reduce theorem \ref{str polyn thm} to a lemma which will allow us to work with one power of the Jacobi operator $L$ at a time.  Define the polynomials $\wt{P}^{(a)}_{\lambda,\eta,J},\wt{P}^{(b)}_{\lambda,\eta,J}$ by
\begin{align}
(L^{J-1})_{n,n} \sim &  \sum_{\lambda,\eta} \wt{P}_{\lambda,\eta,J}^{(a)}\left(s, r \right)    \partial_{n}^{\lambda}s\,  \partial_{n}^{\eta} r, \label{qheic103}\\
(L^{J-1})_{n,n-1} \sim & \sum_{\lambda,\eta} \wt{P}_{\lambda,\eta,J }^{(b)}\left(s,r \right)   \partial_{n}^{\lambda}s\, \partial_{n}^{\eta} r .
\end{align}
Unlike the polynomials $P^{(a \text{ or }b)}_{\lambda,\eta}$, it is clear that the polynomials $\wt{P}^{(a \text{ or }b)}_{\lambda,\eta,J}$ exist.  For example, if we let $J=3$, then (\ref{qheic103}) becomes
\begin{align}
L^{2}_{n,n}=& r_{n+1} + s_{n}^2 +r_{n} \nonumber \\ 
\sim & \left( s^2+2r \right) +N^{-1}\left( 2s s'+ r'\right) \nonumber \\ &+N^{-2}\left( (s')^2+ss'' +\frac{1}{2} r''\right)+ \ldots .
\end{align}
The first several polynomials $\wt{P}^{(a)}_{\lambda,\eta,3}(s,r)$ can be immediately read off from the above example.  For example $\wt{P}^{(a)}_{1,\phi,3}(s,r)=2 s$, and $\wt{P}^{(a)}_{1+1,\phi,3}(s,r)=1$.
\begin{lemma} \label{alt str polyn thm}
For all $\lambda, \eta$ there exist polynomials $P^{(a)}_{\lambda,\eta},P^{(b)}_{\lambda,\eta}$, such that the following identity holds for all $J$:
\begin{align}
\wt{P}^{(a)}_{\lambda,\eta,J}(s,r) =& P^{(a)}_{\lambda,\eta}(\partial_s,\partial_r) \cdot [h^0](h + s+ r h^{-1})^{J-1} \label{univ str poly 125}, \\
\wt{P}^{(b)}_{\lambda,\eta,J}(s,r) =& P^{(b)}_{\lambda,\eta}(\partial_s,\partial_r) \cdot [h^{-1}](h + s+ r h^{-1})^{J-1}. \label{univ str poly 126}
\end{align}
\end{lemma}
We will prove lemma \ref{alt str polyn thm} in section \ref{string general case sect}.
\begin{proof}[Proof that lemma \ref{alt str polyn thm} implies theorem \ref{str polyn thm}]
First consider the string equation $0=V'(L)_{n,n}$.  We must show that the right hand side is as given in theorem \ref{str polyn thm}:
\begin{align*}
 V'(L)_{n,n} 
=& \sum_{J\geq 1} J (t_J+[1\text{ if }J=2]) (L^{J-1})_{n,n} \\
\sim & \sum_{\lambda,\eta}  \partial_{n}^{\lambda}s\,  \partial_{n}^{\eta} r  \sum_{J\geq 1} J (t_J+[1\text{ if }J=2]) \wt{P}^{(a)}_{\lambda,\eta,J}(s,r) \\
=& \sum_{\lambda,\eta}  \partial_{n}^{\lambda}s\, \partial_{n}^{\eta} r  \sum_{J\geq 1} J (t_J+[1 \text{ if }J=2])  \\ 
& \qquad \qquad \times P^{(a)}_{\lambda,\eta}(\partial_{s}, \partial_{r})\cdot [h^0] (h+ s+ r h^{-1}  )^{J-1}\\
=& \sum_{\lambda,\eta}   \partial_{n}^{\lambda}s\, \partial_{n}^{\eta} r  P^{(a)}_{\lambda,\eta}(\partial_{s}, \partial_{r}) \cdot [h^0] V'( h + s + r  h^{-1}) .
\end{align*}
An analogous calculation applies for the other string equation $x=V'(L)_{n,n-1}$.
\end{proof}

\subsection{Motzkin path formulation of string equations} \label{Motzkin sect}
A Motzkin path of length $\ell$ is a function $p$ from $\{0,1,2,\ldots, \ell\}$ to $\mathbb{Z}$ such that $p(m+1)-p(m) \in \{-1,0,1\}$ for each $m$.  We also think of $p$ as the piecewise linear function obtained by joining points on the graph of $p$ with line segments.  
Because the matrix $L$ is tridiagonal, the entries of a power of $L$ can be expressed in terms of a sum over Motzkin paths.  Let $M^{(J-1)}_{0,0}$ be the set of Motzkin paths of length $J-1$ that start and end at height $0$ (that is, $p(0)=p(J-1)=0$).  Then
\begin{align}
L^{J-1}_{n,n} =& \sum_{p\in M^{J-1}_{0,0}} C(p) \label{MPformula134}\\
C(p):=& \prod_{i=0}^{\text{len}(p)-1} 
\begin{cases} 1 & \text{ if }p(i+1)>p(i), \\
s_{n+p(i),N} & \text{ if }p(i+1)=p(i) \\
r_{n+p(i),N} & \text{ if } p(i+1)<p(i)  \end{cases}.\label{MPformula135}
\end{align}
We call the quantity $C(p)$ the contribution of the Motzkin path $p$. 
The string equations can be expressed in Motzkin path form:
\begin{align}
0=& s_n + \sum_j j t_j \sum_{p\in M^{j-1}_{n,n}} C(p), \\
\frac{n}{N} =& r_n +\sum_j  j t_j \sum_{p\in M^{j-1}_{n,n-1}} C(p).
\end{align}
We wish to use the Motzkin path form of the string equations to state a formula for the modified string polynomials $\wt{P}^{(a\text{ or }b)}_{\lambda,\eta,J}$.  
From equation (\ref{qheic103}), which defines the polynomials $\wt{P}^{(a)}_{\lambda,\eta}$, it is clear that there exist polynomials $C_{\lambda,\eta}(p)$ such that
\begin{align}
C(p)  =& \sum_{\lambda,\eta} C_{\lambda,\eta}(p) \partial^{\lambda}_n s\, \partial^{\eta}_n r.
\end{align}
The quantities $ C_{\lambda,\eta}(p)$, which are polynomials in $r$ and $s$, can be easily found by substituting the asymptotic expansions for index-shifted recurrence coefficients and collecting terms.  Thus
\begin{align*}
\sum_{\lambda,\eta} \wt{P}^{(a)}_{\lambda,\eta,J} \partial^{\lambda}_n s\, \partial^{\eta}_n r 
=& L^{j-1}_{n,n} \\
=& \sum_{p \in M^{(J-1)}_{0,0}} C(p) \\
=& \sum_{p \in M^{(J-1)}_{0,0}}  \sum_{\lambda,\eta} C_{\lambda,\eta}(p) \partial^{\lambda}_n s\, \partial^{\eta}_n r ;
\end{align*} 
and therefore we have
\begin{align}
\wt{P}^{(a)}_{\lambda,\eta,J} =& \sum_{p\in M^{J-1}_{0,0}} C_{\lambda,\eta}(p), \label{ModStrPolyn1} \\
\wt{P}^{(b)}_{\lambda,\eta,J} =& \sum_{p\in M^{J-1}_{0,-1}} C_{\lambda,\eta}(p). \label{ModStrPolyn2}
\end{align}
\begin{figure}
\begin{center}
\includegraphics[width=.6\textwidth]{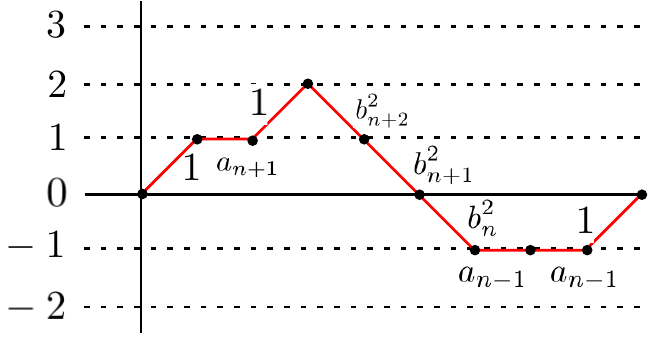}
\end{center} 
\caption{For this Motzkin path $C(p)=a_{n+1}b^{2}_{n+2} b^{2}_{n+1} b^{2}_n a^{2}_{n-1}$.}
\end{figure} 
For example, consider the Motzkin path $p$ shown in figure 1.  The path $p$ contributes $C(p)=s_{n+1}r_{n+2} r_{n+1}r_n s_{n-1}^2$ to the quantity $L^{9}_{0,0}$, calculated as a sum over Motzkin paths according to formula (\ref{MPformula134}).  To find the contribution $C(p;10,1,\phi)$ of this path to the string polynomial $\wt{P}^{(a)}_{2+1,\phi}$, we substitute the series
\begin{align}
s_{n+k,n} \sim& s + \frac{k}{1!}N^{-1}s' + \frac{k^2}{2!}N^{-2} s''+\ldots, \label{10948}\\
r_{n+k,n} \sim& r + \frac{k}{1!}N^{-1}r' + \frac{k^2}{2!}N^{-2} r''+\ldots 
\end{align}
in the formula for $C(p)$ and read off the $N^{-1}s'$ term:
\begin{align*}
C(p)=& s_{n+1} r_{n+2} r_{n+1} r_{n} s_{n-1}^2 \\
=& s^3 r^3 +N^{-1}\left(-s^2 r^3 s' + 3 s^3 r^2 r'\right) +O(N^{-2}) \\
\implies \qquad C_{1,\phi,10}(p )=& -s^2 r^3.
\end{align*}
Notice that in this example, computing the contribution of a path to a string polynomial amounts to selecting one flat step of the Motzkin path from which to choose the $N^{-1}\partial_{x}$ term in expansion (\ref{10948}); from all other steps we choose the leading order term $s$ or $r$.

We now give a formula for modified string polynomials in terms of Motzkin paths.  Define a marking of a Motzkin path $p$ with partitions $\lambda,\eta$ to be a pair of functions $\phi_{\rightarrow},\phi_{\searrow}$, where: $\phi_{\rightarrow}$ assigns to each part of $\lambda$ the left endpoint of a distinct flat step of $p$, and $\phi_{\searrow}$ assigns to each part of $\eta$ the left endpoint of a distinct downward step of $p$.  The contribution of the Motzkin path $p$ to $P^{(a)}_{\lambda,\eta}$ can be computed as a sum over all markings of $p$ with $\lambda$ and $\eta$.  Also, for a path $p$ let $\rightarrow(p),\searrow(p)$ respectively be the number of flat and downward steps of $p$.
\begin{lemma}
\begin{align}
\wt{P}^{(a)}_{\lambda,\eta,J}(s,r) =& \sum_{p\in M^{(J-1)}_{0,0} } \sum_{\phi_{\rightarrow},\phi_{\searrow}} 
s^{\rightarrow(p)} 
r^{\searrow(p)}
\prod_{i} \frac{p(\phi_{\rightarrow}(i))^{\lambda_i}}{\lambda_i !s} 
\prod_{i} \frac{p(\phi_{\searrow}(i))^{\eta_i}}{\eta_i !r}  \label{motzkin 131}
\end{align}
\end{lemma}
\begin{proof}
Notice that the dependence of the right hand side on $\lambda$ and $\eta$ is through the conditions on the markings $\phi_{\rightarrow},\phi_{\searrow}$ explained above.  We compute the modified string polynomial according to formula (\ref{ModStrPolyn1}).  For each downward step of the Motzkin path, we must sum over possible choices of a term in the expansion $r_{n+p(i)}=\sum_m p(i)^m \partial^{m}_n r /m!$.  Similarly for each flat step of the path.  The sums over the markings $\phi_\searrow$ and $\phi_{\rightarrow}$ encode these markings.
\end{proof}

\subsection{Example $P^{(a)}_{\phi,2}$} \label{str polyn example sect}

In this section we give a formula for the string polynomial $P^{(a)}_{\lambda,\eta}$ in the case that $\lambda$ is the empty partition, and $\eta$ has a single part of size $m$.  We will then restrict to the case $m=2$ and compute $P^{(a)}_{\phi,2}$ explicitly.  Our proof of theorem \ref{str polyn thm} in the next section simply generalizes the calculation here.

Appealing to formula (\ref{motzkin 131}), let us identify the conditions on the sum over $\phi_{\rightarrow}$ and $\phi_{\searrow}$.  Since the partition $\lambda=\phi$ is empty, there is no sum over choices of $\phi_{\rightarrow}$.  The function $\phi_{\searrow}$ assigns the single part $m$ of the partition $\eta$ to one of the downward steps of the Motzkin path $p$.  This is the unique downward step $i$ from which we choose the $N^{-m}\partial_{x}^m$ term from the expansion on $r_{n+p(i)}$.   Let $i_1,j_1,k_1$ be the numbers of upward, flat, and downward steps of $p$ before this special step, and $i_2,j_2,k_2$ be the numbers of upward, flat and downward steps after it.  Clearly we have the following relations
\begin{align}
\sideset{}{_{\alpha=1}^2}\sum i_\alpha+ j_\alpha +k_\alpha =& J-2 \\
\sideset{}{_{\alpha=1}^2}\sum i_{\alpha}-k_{\alpha} =& 1
\end{align}
Notice that from any given starting height, the number of Motzkin paths with numbers $i,j,k$ of upward, flat, and downward steps will be given by the trinomial coefficient $\trinom{\ast}{i}{j}{k} =\frac{(i+j+k)!}{i!j!k!}$.
Then the string polynomial can be calculated as follows:
\begin{align}
\wt{P}^{(a)}_{\phi,m}(s,r;J) =& \sum_{\substack{\sum_\alpha i_\alpha+j_\alpha+ k_\alpha =J-2 \\ \sum_\alpha i_\alpha- k_{\alpha} =1}} s^{j_1+j_2} r^{k_1+k_2} \frac{(i_1-k_1)^m}{m!} \nonumber \\
& \qquad \qquad \qquad \qquad \times\trinom{\ast}{i_1}{j_1}{k_1} \trinom{\ast}{i_2}{j_2}{k_2} \label{str coeff 1} .
\end{align}
Notice that the height before the downward step from which we take the $N^{-m}\partial_{x}^m$ term is given by the difference between the numbers $i_1,k_1$ of upward and downward steps.

We now use combinatorial tricks to evaluate the right hand side in formula (\ref{str coeff 1}).  Instead of explicitly putting conditions on the sum, we introduce dummy variables $\ell$ for the length of the path, and $h$ for the change in height of the path.  
\begin{align}
&\wt{P}^{(a)}_{\phi,m}(s,r;J) \nonumber \\
&= [h^{1} \ell^{J-2} ] |_{h_\alpha \rightarrow h } \sum_{i_1,\ldots,k_2\geq 0} \frac{(i_1-k_1)^{m}}{m!}    
 \prod_{\alpha=1}^2 (\ell   h_\alpha)^{i_\alpha} (\ell s)^{j_\alpha} (\ell r h_{\alpha}^{-1})^{k_\alpha} \trinom{\ast}{i_\alpha}{j_\alpha}{k_\alpha}\nonumber  \\
&= [h^{1} \ell^{J-2} ] |_{h_\alpha \rightarrow h } \sum_{i_1,\ldots,k_2\geq 0} \frac{(i_1-k_1)^{m}}{m!}    
 r^{-1/2} \prod_{\alpha=1}^2 (\ell   \sqrt{r} h_\alpha)^{i_\alpha} (\ell s)^{j_\alpha} (\ell \sqrt{r} h_{\alpha}^{-1})^{k_\alpha} \trinom{\ast}{i_\alpha}{j_\alpha}{k_\alpha}\nonumber  \\
&= [h^{1} \ell^{J-2}]|_{h_\alpha \rightarrow h} \frac{(h_1 \partial_{h_1} )^m}{m!}
  \sum_{i_1,\ldots,k_2\geq 0}
 r^{-1/2} \prod_{\alpha=1}^2 (\ell  \sqrt{r} h_\alpha)^{i_\alpha} (\ell s)^{j_\alpha} (\ell \sqrt{r} h_{\alpha}^{-1})^{k_\alpha} \trinom{\ast}{i_\alpha}{j_\alpha}{k_\alpha} \nonumber  \\ 
&=[h^{1} \ell^{J-2} ]|_{h_\alpha \rightarrow h} \frac{(h_1 \partial_{h_1} )^m}{m!} r^{-1/2} \prod_{\alpha=1}^2 \frac{1}{1- \ell( \sqrt{r} h_\alpha + s + \sqrt{r} h_{\alpha}^{-1})}. \label{str coeff ex 1}
\end{align}
The order of operations reads from right to left: apply derivatives, set $h_{\alpha}$'s equal to $h$, then extract the series coefficients.  Let us comment that in formula (\ref{str coeff ex 1}), the coefficient of $\ell^{J-2}$ must be extracted first, then the coefficient of $h^1$, and not vice-versa.  

For any particular value of $m$, we are able to evaluate the derivatives in equation (\ref{str coeff ex 1}) symbolically, then use familiar Taylor series to extract the appropriate coefficients.  Consider the example $m=2$.  Evaluating the two derivatives in (\ref{str coeff ex 1}), we are able to calculate
\begin{align}
\wt{P}^{(a)}_{\phi,2}(s,r;J) =& [h^{1} \ell^{J-2}]
\frac{(h+h^{-1})\ell /2}{(1-\ell s-\ell \sqrt{r} (h+h^{-1}))^3} + \frac{(h^{2}-2+h^{-2})\ell^2 \sqrt{r}}{(1-\ell s-\ell \sqrt{r} (h+h^{-1}))^4} \nonumber \\
=& [h^{1} \ell^{J-2}] \left( \frac{1}{4}\partial_s \partial_{\sqrt{r}} + \frac{1}{6}\sqrt{r} \partial_s \partial_{\sqrt{r}}^2 -\frac{2}{3}\sqrt{r} \partial_{s}^3 \right) \frac{\ell^{-1} }{1-\ell s-\ell \sqrt{r} (h+h^{-1})} \nonumber \\
=& \left( \frac{1}{4} \partial_s \partial_{\sqrt{r}} + \frac{1}{6} \sqrt{r}\partial_s \partial_{\sqrt{r}}^2 -\frac{2}{3}\sqrt{r} \partial_{s}^3 \right) [h^1](\sqrt{r} h+ A + \sqrt{r} h^{-1})^{J-1} \nonumber  \\
=& \left( \frac{1}{6} \partial_{s}^2 + \frac{1}{12} \partial_{r} \right) [h^0 ] (h +s + r h^{-1})^{J-1} \label{a923nv}
\end{align}
We have expressed the coefficient $\wt{P}^{(a)}_{\phi,2,J}(s,r)$ as a differential operator acting on $(h +s + r h^{-1})^{J-1}$, exactly as in equation (\ref{univ str poly 125}).  In this case we see that the differential operator does not depend on $J$; theorem \ref{str polyn thm} asserts that this will be true in general.  We have shown that
\begin{align}
P^{(a)}_{\phi,2} =& \frac{1}{6} \partial_{s}^2 + \frac{1}{12} \partial_{r}.
\end{align}
For the last step in calculation (\ref{a923nv}), we used the following identities:  
\begin{align}
\partial_s [h^1](\sqrt{r}h +s +\sqrt{r}h^{-1})^{J-1} =& \frac{1}{2} \partial_{\sqrt{r}} [h^0](\sqrt{r}h +s +\sqrt{r} h^{-1})^{J-1} \label{lowering1} \\
\partial_{s}^2 [h^0] (\sqrt{r}h + s + \sqrt{r} h^{-1})^{J-1} 
=& \frac{1}{4}\left(r^{-1/2} \partial_{\sqrt{r}} + \partial_{\sqrt{r}}^2\right)[h^0](\sqrt{r}h + s + \sqrt{r}h^{-1})^{J-1}\label{shifting 1}
\end{align}
These identities are somewhat trivial, but we give proofs in section \ref{MotzkinIdentSect}; see lemmas \ref{pathLoweringLemma} and \ref{pathDerivSwapLemma}.

\subsection{General case} \label{string general case sect}

We now extend the method in the above example to compute $P^{(a)}_{\lambda,\eta}(\partial_s,\partial_r)$ for arbitrary partitions $\lambda,\eta$.  Clearly the polynomials $P^{(b)}_{\lambda,\eta}$ can be calculated analogously, but replacing $M^{(J-1)}_{0,0}$ with $M^{(J-1)}_{0,-1}$ in formula (\ref{motzkin 131}).  The following lemma gives an explicit (though cumbersome) formula for the string polynomials:
\begin{lemma} \label{str polyn formula}
Let $l(\lambda)$ and $l(\eta)$ be the numbers of parts of the partitions $\lambda,\eta$.  Let $I$ be the left ideal of the ring $\mathbb{C}[r^{1/2},r^{-1/2},\partial_r,\partial_s]$ generated by $r \partial^{2}_{r}+ \partial_r -  \partial_{s}^2$. Then $P^{(a)}_{\lambda,\eta}(\partial_s,\partial_{r})$ is the unique element of degree at most $1$ in $\partial_{r}$ within the following coset:
\begin{align}
P^{(a)}_{\lambda,\eta}(\partial_s,\partial_r) \equiv \sum_{\vec{v}} C_{\vec{v}}(\lambda,\eta) W_{\vec{v}}(\partial_s,\partial_{\sqrt{r}}) \partial_{s}^{l(\lambda)} \partial_{r}^{l(\eta)}
&\qquad \text{mod }I.
\end{align}
The coefficients $W_{\vec{v}}$ are defined by
\begin{align}
&\prod_{q=0}^{l(\lambda)+l(\eta)}\left[ \ell (h \partial_h)^{v_q} (\Q)^{-1}\right]  \label{2bi8} \\
&\qquad =  W_{\vec{v}}(\partial_{s},\partial_{\sqrt{r}})
\begin{cases}
 \partial_{s}^{l(\lambda)+l(\eta)}  (\Q)^{-1} & \text{ if }\sum v_i\text{ is even} \\
 \partial_{s}^{l(\lambda)+l(\eta)+1} \displaystyle \frac{\sqrt{r}(h-h^{-1})/l(\eta) }{\Q} & \text{ if }\sum v_i \text{ is odd} 
\end{cases} \nonumber
\end{align}
The coefficients $C_{\vec{v}}$ are defined by
\begin{align}
\sum_{m,\sigma} \prod_{q=1}^{l(\lambda)+l(\eta)} \frac{ \left(\sum_{\alpha=0}^{q-1} h_\alpha \partial_{h_\alpha} -\sigma_{\alpha} \right)^{m_q}}{m_q !} 
=& \sum_{\vec{v}} C_{\vec{v}}(\lambda,\eta) \prod_{q=0}^{l(\lambda)+l(\eta)} (h_q \partial_{h_q})^{v_q}.
\end{align}
The conditions on the sum over $m,\sigma$ are that the values of $m:\{1,2,\ldots,l(\lambda)+l(\eta) \}\rightarrow \mathbb{Z}$ are in 1-1 correspondence with the parts of $\lambda$ and $\eta$; and that $\sigma:\{1,2,\ldots,l(\lambda)+l(\eta) \}\rightarrow \{0,1\}$ is such that if $\sigma(i)=0$, then $m(i)$ is a part of $\lambda$; and if $\sigma(i)=1$, then $m(i)$ is a part of $\eta$.
\end{lemma}

%
As in section \ref{str polyn example sect} we start with the Motzkin path formula (\ref{motzkin 131}) for the modified string polynomial $\wt{P}^{(a)}_{\lambda,\eta,J}$.
To generalize the calculation there, we let $i_\alpha,j_\alpha,k_\alpha$ be the numbers of upward, flat and downward steps between the $\alpha^{th}$ and $\alpha+1^{st}$ marked steps.  To give a formula analogous to equation (\ref{str coeff 1}), we need to define two auxiliary functions that describe which parts of the partitions $\lambda,\eta$ correspond to which marked steps.
Let $\sigma_p$ be the function function that takes the value $1$ if the $p^{th}$ marked step is a downward step and zero otherwise.  The function $m$ represents an assignment of parts of the partitions $\lambda,\eta$ to the marked steps.
\begin{align}
\wt{P}^{(a)}_{\lambda,\eta,J}(s,r) =&  \sum_{m,\sigma}
\sum_{\substack{\sum i_q -k_q =l(\eta) \\ \sum i_q +j_q +k_q = J-1-l(\lambda)-l(\eta)}}
\prod_{q=1}^{l(\lambda)+l(\eta)} \frac{\left(\sum_{\alpha=0}^{q-1} i_\alpha-k_\alpha -\sigma_\alpha \right)^{m_q}}{m_q !}
\nonumber \\
& \qquad \qquad \qquad \qquad \qquad \;\;\;\times \prod_{q=0}^{l(\lambda)+l(\eta)} s^{i_q} r^{k_q} \trinom{\ast}{i_q}{j_q}{k_q}  \label{strcoeff 34}
\end{align}
The conditions on the sum over $m,\sigma$ are that the values of $m:\{1,2,\ldots,l(\lambda)+l(\eta) \}\rightarrow \mathbb{Z}$ are in 1-1 correspondence with the parts of $\lambda$ and $\eta$; and that $\sigma:\{1,2,\ldots,l(\lambda)+l(\eta) \}\rightarrow \{0,1\}$ is such that if $\sigma(i)=0$, then $m(i)$ is a part of $\lambda$; and if $\sigma(i)=1$, then $m(i)$ is a part of $\eta$.  Notice that the principal way in which the right hand side of formula (\ref{strcoeff 34}) depends of $\lambda,\eta$ is through the conditions on the sum over $m,\sigma$.

By a calculation analogous to the example in section \ref{str polyn example sect}, we find that
\begin{align}
\wt{P}^{(a)}_{\lambda,\eta}(s,r) =&
\sum_{m,\sigma} [h^{l(\eta)} \ell^{J-1-l(\lambda)-l(\eta)}]\bigg|_{h_\alpha \rightarrow h} \frac{1}{r^{-l(\eta)/2}} \prod_{q=1}^{l(\lambda)+l(\eta)} \frac{ \left(\sum_{\alpha=0}^{q-1} h_{\alpha} \partial_{h_\alpha} -\sigma_\alpha \right)^{m_q}}{m_q !}
\nonumber \\
& \qquad \qquad \qquad \qquad \qquad \times \prod_{q=0}^{l(\lambda)+l(\eta)} \frac{1}{1-\ell ( \sqrt{r} h_q + s + \sqrt{r} h_{q}^{-1})}
\end{align}
The dependence of terms in the above sum on $m$ and $\sigma$ appears only through the differential operators acting on the product of generating functions.  We can thus expand this differential operator as a polynomial in $\partial_{h_1},\ldots,\partial_{h_{l(\lambda)+l(\eta)}}$, defining the symbol $C_{\vec{v}}(\lambda,\eta)$ for the coefficients of this polynomial; we have chosen to explicitly indicate that the coefficients depend on $\lambda,\eta$.
\begin{align*}
\sum_{m,\sigma} \prod_{q=1}^{l(\lambda)+l(\eta)} \frac{ \left(\sum_{\alpha=0}^{q-1} h_{\alpha} \partial_{h_\alpha} -\sigma_\alpha \right)^{m_q}}{m_q !} =& \sum_{\vec{v}} C_{\vec{v}}(\lambda,\eta) \prod_{q=0}^{l(\lambda)+l(\eta)} (h_q \partial_{h_q})^{v_q}
\end{align*}
The sum is over tuples $\vec{v}=(v_0,v_1,\ldots,v_{l(\lambda)+l(\eta)})$, but clearly only finitely many terms in the sum are nonzero.  In terms of the coefficients just defined, we find the following formula.
\begin{align*}
\wt{P}^{(a)}_{\lambda,\eta,J}(s,r) =& [h^l(\eta) \ell^{J-1-l(\lambda)-l(\eta)}] \frac{1}{r^{l(\eta)/2}}  \sum_{\vec{v}} C_{\vec{v}}(\lambda,\eta) \\
& \qquad \qquad \qquad \qquad  \prod_{q=0}^{l(\lambda)+l(\eta)} \left[(h \partial_h)^{v_q} \frac{1}{1-\ell( \sqrt{r} h + s + \sqrt{r} h^{-1})}\right]
\end{align*}
Following the example in section \ref{str polyn example sect}, we would like to represent this as a differential operator given in terms of $\partial_r$ and $\partial_s$ acting on $1/[1-\ell( \sqrt{r} h + s + \sqrt{r} h^{-1})]$.  The purpose of the following lemma is to exchange $\partial_h$ derivatives for $\partial_s$ and $\partial_s$ derivatives.\\
\begin{lemma} \label{Wr lemma}
For each tuple of nonnegative integers $\vec{v}$, then there exists a polynomial $W_{\vec{v}}(\partial_s,\partial_{\sqrt{r}})$ with coefficients in the polynomial ring $\mathbb{C}[\sqrt{r}]$ such that
\begin{align}
&\prod_{q=0}^{\alpha+\beta}\left[ (h \partial_h)^{r_q} \frac{1}{\Q }\right]  \label{2bi8} \\
&\qquad = \ell^{-l(\lambda)-l(\eta)} W_{\vec{v}}(\partial_{s},\partial_{\sqrt{r}})
\begin{cases}
 \partial_{s}^{l(\lambda)+l(\eta)} \displaystyle  \frac{1}{\Q} & \text{ if }\sum v_i\text{ is even} \\
 \partial_{s}^{l(\lambda)+l(\eta)+1}   \displaystyle \frac{\sqrt{r} (h-h^{-1})/l(\eta) }{\Q} & \text{ if }\sum v_i \text{is odd} 
\end{cases} \nonumber
\end{align}
\end{lemma}
\begin{proof}
First we claim that for any nonnegative integer $v$ there exist rational functions $R_{v,p}(h),\; p=0\ldots v$ with integer coefficients such that
\begin{align}
(h\partial_h)^v \frac{1}{\Q }=&
\sum_{p=0}^v \frac{(\ell \sqrt{r})^p R_{v,p}(h)}{(\Q)^{p+1}}.
\end{align}
In fact furthermore we claim that if $v$ is even, then $R_{v,p}$ is a polynomial function of $h+h^{-1}$; and if $v$ is odd, then $R_{v,p}$ equals $h-h^{-1}$ times a polynomial function of $h+h^{-1}$.
It is easy to verify this claim by induction on $v$.  
It follows that
\begin{align}
\prod_{q=0}^{l(\lambda)+l(\eta)} (h\partial_h)^{r_q} \frac{1}{\Q} =& \sum_{p=0}^{\sum v_i} \frac{ (\ell B)^p \wt{R}_{\vec{r},p}(h)}{(\Q)^{p+l(\lambda)+l(\eta)}},
\end{align}
where $\wt{R}_{\vec{v},p}$ is a polynomial function of $h+h^{-1}$ is $\sum v_i$ is even, and $h-h^{-1}$ times such a function otherwise.  The lemma now follows from because
\begin{align}
\partial_{s}^a \partial_{\sqrt{r}}^b \frac{1}{\Q} =& \frac{\ell^{a+b} (h+h^{-1})^b}{(\Q)^{1+a+b}}.
\end{align}
We note that the asserted divisibility by $\partial_{s}^{l(\lambda)+l(\eta)-1}$ is because $\partial_s$ is needed to increase the exponent in the denominator without increasing the degree of the numerator.
\end{proof}

It is convenient to express the right hand side of (\ref{2bi8}) in this form because the following identity will allow us to recombine the even and odd $\sum v_i$ cases in the next calculation.
\begin{align}
\sqrt{r}\partial_s [h^p](h-h^{-1})(\sqrt{r} h + s +\sqrt{r} h^{-1})^{v} =& p [h^p](\sqrt{r} h + s +\sqrt{r} h^{-1})^{v} \label{qniv}
\end{align}
Formula (\ref{qniv}) is a special case of the ``integration by parts'' lemma proven in section \ref{MotzkinIdentSect}.  It follows that
\begin{align}
&\wt{P}^{(a)}_{\lambda,\eta,J}(s,r) \nonumber \\
&=[h^l(\eta) \ell^{J-1-l(\lambda)-l(\eta)}] \frac{1}{r^{l(\eta)/2}} \bigg( \nonumber \\
&\qquad  \sum_{\vec{v}: \;\sum v_i \text{ even}} C_{\vec{v}}\, \ell^{-l(\lambda)-l(\eta)} W_{\vec{v}} (\partial_s, \partial_{\sqrt{r}})\partial_{s}^{l(\lambda)+l(\eta)} \frac{1}{\Q} \nonumber \\
& \qquad+  \sum_{\vec{v}: \; \sum v_i \text{ odd}} C_{\vec{v}}\, \ell^{-l(\lambda)-l(\eta)} W_{\vec{v}} (\partial_s, \partial_{\sqrt{r}}) \partial_{s}^{l(\lambda)+l(\eta)+1} \frac{\sqrt{r}(h-h^{-1})/l(\eta) }{\Q} \bigg) \nonumber \\
&= \frac{1}{r^{l(\eta)/2}} \sum_{\vec{v}} C_{\vec{v}}\, W_{\vec{v}}(\partial_s, \partial_{\sqrt{r}})\partial_{s}^{l(\lambda)+l(\eta)} [h^{l(\eta)}] (\sqrt{r} h + s + \sqrt{r} h^{-1})^{J-1} \nonumber \\
&=  \sum_{\vec{v}} C_{\vec{v}}\, W_{\vec{v}}(\partial_s, \partial_{\sqrt{v}})\partial_{s}^{l(\lambda)}\partial_{r}^{l(\eta)} [h^0] (\sqrt{r} h + s + \sqrt{r} h^{-1})^{J-1} \label{1836} \\
&= P^{(a)}_{\lambda,\eta}(\partial_s,\partial_r) \, [h^0](\sqrt{r} h + s + \sqrt{r}h^{-1})^{J-1}. \label{1837}
\end{align}
To get (\ref{1836}) we applied the ``repeated path lowering identity'' of section \ref{MotzkinIdentSect}.  Clearly, there exists $\wt{P}$ such that the final expression (\ref{1837}) of our formula for string polynomials holds; however by the ``derivative swapping identity'' of section \ref{MotzkinIdentSect}, we can and do choose $\wt{P}$ to have polynomial degree at most one in $\partial_r$.  This completes the proof of lemma \ref{alt str polyn thm}, and thus also theorem \ref{str polyn thm}.

\section{Solution of string equations}

The following valence independent formula was known to Bessis et. al. \cite{BIZ}:
\begin{align}
F^{(0)} =& \int_{0}^{x} \int_{0}^{x_1} \log \frac{z}{x_2} \, \d x_2 \, \d x_1. \label{e0ValIndep}
\end{align}
We are not aware of any closed form evaluation of the integrals in (\ref{e0ValIndep}) which holds for general asymmetric polynomial potentials.  However in the case of a potential $V(\lambda)=\frac{1}{2}\lambda^2 + t \lambda^j$, where $j$ is allowed to be odd, we have found an explicit formula \cite{PatricksDissertation}.  We consider (\ref{e0ValIndep}) to be the valence independent formula for $F^{(0)}$; thus the real work begins at genus 1.  In this section we compute $F^{(1)}$ using the string equations.

\subsection{The unwinding identity} \label{unwinding sect}
Define two sequences of functions by
\begin{align}
\vect{\phi_m}{\psi_m} :=& [h^0]V^{(m+1)} (h +u+z h^{-1}) \vect{1}{h } \label{phi psi 1} \\
=& \vect{\partial_{u}^{m+1}}{z \partial_{u}^{m} \partial_{z}} [h^0] V(h+u+zh^{-1}).
\end{align}
Notice that the above expressions depend explicitly on the potential $V$.  The functions $\phi_m$ and $\psi_m$ are useful because it turns out that they can actually be expressed in a form independent of the potential (``valence independent'') are furthmore satisfy the ``unwinding identity''. 

\begin{lemma}[Valence independence of $\phi_m,\psi_m$ and the unwinding identity]\label{unwinding lemma}
The functions $\phi_m,\psi_m$ are determined by the following recurrence
\begin{align}
 \vect{\phi_{m+1}}{\psi_{m+1}} =&  \frac{1}{(z')^2-z (u')^2} \matr{-z u' }{z'}{z z'}{-z u'} \partial_x  \vect{\phi_m}{\psi_m}, \label{recursion}
\end{align}
along with initial conditions $\phi_0=0$ and $\psi_0 =x$.
For integers $m\geq 1$, we have 
\begin{align}
z' \phi_m + u' \psi_m =&  \psi_{m-1}'  \label{index lowering 1} \\
z u' \phi_m + z' \psi_m =& z  \phi'_{m-1} \label{index lowering 2} 
\end{align}
\end{lemma}
In the case $m=1$ the right hand sides of the unwinding identities are $\psi_{0}'=1$ and $z\phi'_{0}=0$.  We will use this in section \ref{solving string} to isolate highest genus terms when solving the string equations.  The case $m=1$ was used in \cite{EP11}, and the proof we give below essentially parallels the calculation there; lemma \ref{unwinding lemma} extends Ercolani and Pierce's idea to $m>1$.  
\begin{proof}
We proceed by induction on $m$.  The base case $m=0$ is the leading order equation of the continuum string hierarchy.  Let $Y=h+u+zh^{-1}$.  
\begin{align}
 \vect{ \partial_x\phi_{m}}{\partial_x\psi_{m}} =& [h^0] V^{(m+2)}(Y) \matr{1}{h^{-1}}{h}{1} \vect{u'}{z'} \label{rec proof 1}
\end{align}
By elementary linear algebraic manipulations, we see that (\ref{rec proof 1}) is equivalent to the following system; we then apply the path reflecting identity, which is lemma \ref{path reflect lemma} of appendix \ref{MotzkinIdentSect}.
\begin{align}
 \vect{ \partial_x\psi_{m} }{z \partial_x \phi_m } =& [h^0] V^{(m+2)}(Y) \matr{1}{z^{-1}h }{zh^{-1} }{1} \vect{z'}{z u'} \nonumber \\
 =& [h^0] V^{(m+2)}(Y) \matr{1}{h^{-1} }{h }{1} \vect{z'}{z u'} \label{rec proof 2}
\end{align}
We now combine (\ref{rec proof 1}) and (\ref{rec proof 2}) to form a matrix equation; once again using the path reflecting identity we find that
\begin{align}
\matr{\partial_x \phi_m}{ \partial_x \psi_m}{\partial_x \psi_m}{z \partial_x \phi_m } 
=&  \matr{u'}{ z^{-1} z'}{z'}{u'}
[h^0] V^{(m+2)}(Y) \matr{1}{h}{h }{z}  \label{proof eq 72}
\end{align}
The first column of the above equation is equivalent to the recurrence (\ref{recursion}) for $\phi_m,\psi_m$.  Rearranging equation (\ref{proof eq 72}) and expressing the series coefficients in terms of $\phi_m$ and $\psi_m$  gives the unwinding identity.
\end{proof}

\subsection{Extraction of generating functions}
So far it has been useful to work in terms of the asymptotic expansions $s= \sum N^{-g} u_g$ and $r = \sum N^{-2g}z_g$.  But to solve for $z_g,u_{2g}$ we will need to extract terms from these expansions.  For this purpose, we have the following formula.  
\begin{lemma}
If $G(s,r)$ is a polynomial, then
\begin{align}
[N^{-p}] G(s,r) =& \sum_{\lambda, \eta:\; \sum \lambda_i + 2 \eta_i =p}  \prod_i \frac{u_{\lambda_i}  z_{\eta_i}}{\lambda_i ! \eta_i !} [N^0] \partial_{s}^{l(\lambda)} \partial_{r}^{ l(\eta) } G(s,r)
\end{align}
\end{lemma} \label{genus extract lemma}
\begin{proof} Consider the case where $G$ is a monomial $s^p r^q$; we have a contribution corresponding to each pair tuples $(g_\alpha)_{\alpha=1}^p$ and $(g'_\alpha)_{\alpha=1}^q$ which uniquely determine a way of choosing a terms $u_{g_\alpha} N^{-g_\alpha}$ from each factor $s$ and $z_{g'_{\alpha}} N^{-2g'_{\alpha}}$ from each factor $r$.  The partitions $\lambda,\eta$ are the multisets corresponding to the tuples $g,g'$.  The number of tuples corresponding to a partition in this way can be expressed in terms of factorials in such a way that we arrive at the lemma statement.
\end{proof}
\noindent The sum is over partitions $\lambda$ and $\eta$, and by $l(\lambda)$ we mean the number of parts of the partition $\lambda$.  Thus we are able to calculate that
\begin{align}
[N^{-1} h^0] V'(\mpsym) =& u_1 \partial_s [h^0] V'(\mpsym) \nonumber \\
=& u_1 \phi_1, 
\end{align}
and also that
\begin{align}
[N^{-2} h^0] V'(\mpsym) =& \left( z_1 \partial_{r} + u_2 \partial_s + \frac{1}{2} u_{1}^2 \partial_{s}^2 \right) [h^0]V'(\mpsym) \nonumber \\
=& z_1 \psi_1/z + u_2 \phi_1 + \frac{1}{2}u_{1}^2 \phi_2.
\end{align}

\subsection{Example: solution at order $N^{-2}$} \label{string e1 sect}
Using our table of string polynomials, we can now write down the equations at order $N^{-2}$ from the continuum string hierarchy.  Consider the first string equation $0=V'(L)_{n,n}$; at order $N^{-2}$ we have
\begin{align}
0= & \bigg\{ [N^0] \bigg( s'' \frac{1}{6} r \partial_s \partial_{r} 
+ (s')^2 \frac{1}{12} r \partial_{s}^2  \partial_{r} 
+ s' r' \frac{1}{6} \partial_{s}^3 \nonumber \\
&\qquad + r'' \left( \frac{1}{6} \partial_{s}^2 + \frac{1}{12}  \partial_{r} \right) 
+ (r')^2 \left( \frac{1}{12} r^{-1} \partial_{s}^2 - \frac{1}{12} r^{-1} \partial_{r}  + \frac{1}{12} \partial_{s}^2 \partial_{r} \right) \bigg) \nonumber \\
& \qquad + [N^{-1}] r' \frac{1}{2} \partial_{r} + [N^{-2}] \bigg\}\;\cdot [h^0]V'(\mpsym) \nonumber \\
=& u'' \frac{1}{6} \psi_2 
+ (u')^2 \frac{1}{12} \psi_3 
+ u' z' \frac{1}{6} \phi_3 \nonumber \\
&\;\;+ z''\left( \frac{1}{6} \phi_2 + \frac{1}{12} z^{-1}\psi_1 \right) 
+ (z')^2 \left( \frac{1}{12} z^{-1} \phi_2 - \frac{1}{12} z^{-2} \psi_1 + \frac{1}{12} z^{-1} \psi_3 \right)\nonumber  \\
&\;\; + z' \frac{1}{2} z^{-1} u_1 \psi_2 + z_1 z^{-1} \psi_1 + u_2 \phi_1 + u_{1}^2 \frac{1}{2} \phi_2 \label{str 2s}
\end{align}
Proceeding similarly with the second string equation $x=V'(L)_{n,n-1}$ we obtain
\begin{align}
0=&  u'' \left(\frac{1}{6}z \phi_2 + \frac{1}{12} \psi_1\right) + (u')^2 \left( \frac{1}{12} z \phi_3 + \frac{1}{12} \psi_2 \right) + u' z' \frac{1}{6} \psi_3 \nonumber  \\
& \qquad + z'' \frac{1}{6}\psi_2 + (z')^2 \left( \frac{1}{12} \phi_3 -\frac{1}{12}z^{-1} \psi_2 \right) \nonumber  \\
& \qquad +u_{1}'\left(-\frac{1}{2}\psi_1 \right)+u' \left(-\frac{1}{2} u_1 \psi_2 \right) + u_2 \psi_1 + \frac{1}{2}u_{1}^2 \psi_2 + z_1 \phi_1 \label{str 2d}
\end{align}
We can solve for $z_1$ in terms of lower genus quantities by taking the linear combination of equations $ zu' (\ref{str 2s}) + z' (\ref{str 2d})$; this combination allows us to isolate $z_1$ using identity $\ref{index lowering 1})$.  We obtain
\begin{align}
-z_1=&\psi_3 \left(\frac{1}{4} u' \left(z'\right)^2+\frac{1}{12} z \left(u'\right)^3\right)
+\phi_3 \left(\frac{1}{4} z \left(u'\right)^2
   z'+\frac{\left(z'\right)^3}{12}\right) \nonumber \\
&+\psi_2 \left(\frac{1}{12} \left(u'\right)^2 z'+\frac{1}{6} z u' u''+\frac{1}{2} u_1^2 z'-\frac{\left(z'\right)^3}{12 z}+\frac{z' z''}{6}\right) \nonumber \\
&+\phi_2 \left(\frac{1}{6} z u'' z'+\frac{1}{2} u_1^2 z
   u'+\frac{1}{6} z u' z''+\frac{1}{12} u' \left(z'\right)^2\right) \nonumber \\
& +\psi_1 \left(\frac{u'' z'}{12}+\frac{u' z''}{12}-\frac{u' \left(z'\right)^2}{12 z}-\frac{1}{2} u_1' z'\right)  \label{z1 1}
\end{align}
The appearantly complicated expression above simplifies in a fantastic way due to identities (\ref{index lowering 1}) and (\ref{index lowering 2}), which ``unwind'' the recurrence defining $\phi_m$ and $\psi_m$.  We would like to calculate $F^{(1)}$, which is related to $z_1$ by $F^{(1)} = \partial_{x}^{-2}(z_1/z)-\frac{1}{12} \log (z/x)$.  It would therefore be nice if (\ref{z1 1})$/z$ could be grouped as a double exact derivative.  Toward this purpose we recognize the following simplifications of blocks of terms that will appear in the equation (\ref{z1 1})$/z$:
\begin{align}
&\psi_3 \left(\frac{u' \left(z'\right)^2}{4z} +\frac{1}{12} \left(u'\right)^3 \right)
+\phi_3 \left(\frac{1}{4}  \left(u'\right)^2
   z'+\frac{\left(z'\right)^3}{12z}\right) \nonumber \\
&\qquad +\phi_2 \left( \frac{1}{6}  u'' z' + \frac{1}{6}  u' z'' \right)
+ \psi_2 \left( -\frac{\left(z'\right)^3}{12 z^2}+\frac{z' z''}{6z}+\frac{1}{6}  u' u'' \right) \nonumber \\
&= \partial_x \left( \frac{1}{6} \phi_2 u' z' + \psi_2 \left( \frac{1}{12}(u')^2 + \frac{(z')^2}{12z} \right)\right) \\
&= \partial_x \left( \frac{1}{12}\psi'_{1}u' + \frac{1}{12}\phi'_1 z' \right)
\end{align}
This simplification was achieved by applying identities (\ref{index lowering 1}) and (\ref{index lowering 2}) to the $\phi_3$ and $\psi_3$ terms; collecting exact derivatives; then once again applying (\ref{index lowering 1}) and (\ref{index lowering 2}).  Similarly,
\begin{align}
& \psi_2\left( \frac{u_{1}^2 z'}{2z} + \frac{(u')^2 z'}{12 z} \right) +\phi_2 \left( \frac{1}{2} u_1^2  u' + \frac{ u'(z')^2}{12z} \right) \nonumber \\
& \qquad + \phi_1 u_1 u_{1}' + \psi_1 \left( \frac{u'z''}{12 z} + \frac{u''z'}{12z} - \frac{u'(z')^2}{12 z^2} \right) \nonumber \\
& = \partial_x \left( \frac{1}{2} \phi_1 u_{1}^2 + \psi_1 \frac{u'z'}{12z} \right) \nonumber 
\end{align}
Making these simplifications in (\ref{z1 1}); then using $u_1 = \frac{1}{2}u'$ and yet again identity (\ref{index lowering 2}), we obtain
\begin{align}
-\frac{z_1}{z}=& \partial_x \left( \frac{1}{12}\psi'_{1}u' + \frac{1}{12}\phi'_1 z' \right)
+\partial_x \left( \frac{1}{2} \phi_1 u_{1}^2 + \psi_1 \frac{u'z'}{12z} \right)
- \phi_1 u_1 u'_{1} - \psi_{1} \frac{z' u_{1}'}{2z} \nonumber \\
=& \partial_x \left( -\frac{1}{12} \psi_1 u'' -\frac{1}{12}\phi_1 z'' +\frac{1}{2} \phi_1 u_{1}^2 + \psi_1 \frac{u'z'}{12z} \right) \\
=& \partial_x \left( \phi_1 \left( \frac{(u')^2}{24} - \frac{z''}{12} \right) -\psi_1 \frac{u''}{12} \right) \nonumber \\
=& \frac{1}{24} \partial_x  \frac{(u')^2 z' - 2z' z'' + 2 z u' u''}{(z')^2-z(u')^2}  \nonumber \\
=& -\frac{1}{24} \partial_{x}^2 \log \left[ (z')^2 -z (u')^2 \right]
\end{align}
Finally we obtain the generating function for maps of genus 1:
\begin{align}
F^{(1)} =& \partial_{x}^{-2} \frac{z_1}{z} -\frac{1}{12} \log \frac{z}{x} \nonumber \\
=&\frac{1}{24} \log \left[ (z')^2 -z (u')^2 \right] -\frac{1}{12} \log\frac{z}{x}.
\end{align}

\subsection{Recursive calculation of $z_g$ and $u_g$}
From lemma \ref{genus extract lemma}, it is clear that the highest genus terms in the string equations at order $N^{-2g}$ will be
\begin{align*}
\phi_1 u_{2g} + z^{-1}\psi_1 z_g \qquad \text{in equation }0=[N^{-2g}]V'(L)_{n,n} \\
\psi_1 u_{2g} + \phi_1 z_g \qquad \text{in equation }0=[N^{-2g}]V'(L)_{n,n-1} 
\end{align*}
Thus by the unwinding identity we can take combinations of equations to solve for the highest genus terms.
\begin{align}
z u' [N^{-2g}]V'(L)_{n,n} + z'[N^{-2g}]V'(L)_{n,n-1} \qquad \text{isolates }z_g \label{iso 036}\\
z' [N^{-2g}]V'(L)_{n,n} + u'[N^{-2g}]V'(L)_{n,n-1} \qquad \text{isolates }u_{2g} \label{iso 037}
\end{align}
To calculate expressions for the functions $u_{2g+1}$ in terms of lower genus quantities we do not need the string equations.  It was shown in section 4.2.3 of \cite{PatricksDissertation} that
\begin{align*}
u_{2g+1} =& -\sum_{m=1}^{2g+1} \frac{B_m}{m!} \partial_{x}^m u_{2g+1-m},
\end{align*}
and therefore the odd index functions $u_{2g+1}$ are actually redundant in the sense that by applying the above identity several times $u_{2g+1}$ can be expressed purely in terms of lower genus $u_{2g}$ with even index.  (See \cite{PatricksDissertation}) for an explanation of the relation between $u_{2g}$ and generating functions for maps with legs.)

This proves that the functions $z_g$ and $u_{2g}$ can be recursively calculated from the string equations.

\begin{definition}
The polynomial degree of $z$ or any $x$-derivative $z^{(m)}$ is $1$.  The polynomial degree of $u$ or any $x$-derivative $u^{(m)}$ is $\frac{1}{2}$.  The differential weight of any $x$-derivative $u^{(m)}$ or $z^{(m)}$ is $m$.  
\end{definition}
\begin{definition}
A polynomial in $u,z',u',z',u'',\ldots$ is homogeneous if all its terms have the same polynomial degree and the same differential weight.  The polynomial degree of a product (or quotient) of homogeneous polynomials is given by adding (or subtracting) the polynomial degrees, and similarly for differential weights of products and quotients.
\end{definition}
Notice that the polynomial degree and differential weight are only defined for the subset of the rational functions in $u,z',u',z',u'',\ldots $ satisfying a homogeneity property.  Recall from \cite{EMcLP08,EP11,PatricksDissertation} the scaling relations
\begin{align}
z_g(x,t)=&x^{1-2g}z_g\left(1,(x^{j/2-1}t_j)_{j=1}^d\right), \label{zScale}\\
u_g(x,t)=&x^{1/2-g}u_g\left(1,(x^{j/2-1}t_j)_{j=1}^d\right) \label{uScale}.
\end{align}
A quantity with polynomial degree $p$ and differential weight $m$ has scaling structure $ x^{p-m} G\left( (z^{j/2-1}t_j)_{j=1}^d\right)$.  Thus if one knows any two of the quantities polynomial degree, differential weight, and scaling structure, then the third is uniquely determined.
\begin{theorem} \label{zg structure thm1}
From the string equations, we deduce that there exist formulas of the following forms for the functions $z_g$ and $u_g$:
\begin{align}
z_g =& \frac{P_{z_g}(u,z,u',z',u'',\ldots)}{\left( (z')^2-z (u')^2\right)^{8g-3}} \\
u_{2g} =& \frac{P_{u_{2g}}(u,z,u',z',u'',\ldots)}{\left( (z')^2-z (u')^2\right)^{8g-3}} \\
u_{2g+1} =& \frac{P_{u_{2g+1}}(u,z,u',z',u'',\ldots)}{\left( (z')^2-z (u')^2\right)^{8g-2}} 
\end{align}
where the functions $P_{z_g}$, $P_{u_{2g}}$ and $P_{u_{2g+1}}$ are polynomials in $x$-derivatives of $u$ and $z$, homogeneous with respect to polynomial degree and differential weight.  The differential weight of $z_g$ is $2g$, and the polynomial degree of $z_g$ is $1$.  The differential weight of $u_{2g}$ is $2g$ and the polynomial degree of $u_{2g}$ is $\frac{1}{2}$.  The differential weight of $u_{2g+1}$ is $2g+1$ and the polynomial degree of $u_{2g+1}$ is $\frac{1}{2}$.
\end{theorem}
Let us comment that the above theorem statement is not sharp; consider for example the case of $z_1$: the formula we obtained has a denominator $((z')^2-z(u')^2)^2$, but the theorem statement gives $((z')^2-z(u')^2)^5$.  This does not contradict the theorem statement since the numerator and denominator are allowed to have common factors.  This reason for the failure of sharpness is that each application of the unwinding identity cancels a common factor of $(z')^2-z(u')^2$ from numerator and denominator.  From the calculation of $F^{(1)}$ in section \ref{string e1 sect} we see that the possibility of applying the unwinding identity is sensitive to the exact values of the string polynomials; since our formula for string polynomials is quite complicated, significant new ideas would be required to prove a sharp version of theorem \ref{zg structure thm1}.  

\begin{proof}[Proof of theorem \ref{zg structure thm1}]
The case $g=1$ is clear from the explicit calculation in section \ref{string e1 sect}.  Proceed by induction on $g$, assuming the theorem statement for all genus values less than $g$.
First consider the right hand sides of the string equation $0=[N^{-2g}]V'(L)_{n,n}$.  Expressed in terms of the contributions of Motzkin paths, this is
\begin{align}
[N^{-2g}]V'(L)_{n,n}=& \sum_{j=1}^{d} (j t_j+\delta_{j=2}) \sum_{p\in M^{(j-1)}_{0,0}} [N^{-2g}] C(p) \label{anrp}
\end{align}
where $C(p)$ is defined in equation \ref{MPformula135}.  From the scaling relations (\ref{zScale}-\ref{uScale}), it is clear that every term in this sum has scaling structure
\begin{align}
x^{1/2-2g}G\left((x^{j/2-1}t_j)_{j=1}^d \right). \label{proof scale1}
\end{align}
We express equation (\ref{anrp}) in the form guaranteed by theorem \ref{str polyn thm}, then reduce all odd index functions $u_{2h+1}$ in terms of even index functions $u_{2h'}$, and in each term apply one of the identities
\begin{align*}
[N^0 h^0] \partial_{s}^m V'(h+s+ r h^{-1}) =& \phi_{m+1} \\
[N^0 h^0] \partial_{s}^{m-1} \partial_{\sqrt{r}} V'(h+s+ r h^{-1}) =& \psi_{m+1}. 
\end{align*}
The quantity $[N^{-2g}] V'(L)_{n,n}$ is a sum of terms of the form
\begin{align}
& u_{2g_1}\ldots u_{2g_k} z_{g_{k+1}} \ldots z_{g_{K}} 
u^{(m'_1)}_{2g'_1} \ldots u^{(m'_{k'})}_{2g'_{k'}} 
z^{(m'_{k'+1})}_{g'_{k'+1}} \ldots z^{(m'_{K'})}_{g'_{K'}} \\
&\;\;\times u^{(m_1)} \ldots u^{(m_l)} z^{(m_{l+1})} \ldots z^{(m_L)}
z^{r} \left( \phi_m \text{ or }\psi_m \right), \label{crazy term 302}
\end{align}
where 
\begin{align*}
2g=\sum 2g_i +2g'_i +m_i +m'_i \label{constraint 254}.
\end{align*}

 
The highest genus terms are
$
u_{2g}\phi_1 +z_g z^{-1} \psi_1,
$
but our inductive hypotheses apply to all other terms, which we will call forcing terms.  Each forcing term has the scaling structure in equation (\ref{proof scale1}), and by our inductive hypothesis, differential weight $2g-1$; also it is clear from the recurrence (\ref{recursion}) that both $\phi_m$ and $\psi_m$ have differential weight $-1$.  
Thus it follows that each term has polynomial degree $-\frac{1}{2}$.  One can see by the same argument that all forcing terms in the equation $0=[N^{-2g}]V'(L)_{n,n-1}$ have polynomial degree zero and differential weight $2g-1$.  By taking the linear combinations of equations given by (\ref{iso 036}), (\ref{iso 037}) we have that $z_g$ and $u_g$ are homogeneous with differential weight and polynomial degree as given in the theorem.

It remains only to bound the power of $(z')^2-z(u')^2$ appearing in denominators.  
From the equations at the end of section \ref{string general case sect} we have
\begin{align*}
P^{(a)}_{\lambda,\eta} (\partial_s,\partial_r)\equiv \sum_{\vec{v}} C_{\vec{v}} W_{\vec{v}}(\partial_s,\partial_{\sqrt{r}}) \partial_{s}^{l(\lambda)}\partial_{r}^{l(\eta)} & \text{mod }I,
\end{align*}
and the ``mod $I$'' cannot increase the total degree in $\partial_{s},\partial_r$.
 The sum is indexed by a tuple $\vec{v}$, where the total $|\vec{v}|:=\sum_{i} v_i$ is less than or equal to $|\lambda| + | \eta|$.  
From the proof of lemma \ref{Wr lemma}, it is clear that the degree of the polynomial $W_{\vec{v}}$ is at most $|\vec{v}|$ if this quantity is even, and $|\vec{v}|-1$ otherwise.  Therefore the degree of $P^{(a)}_{\lambda,\eta}(\partial_s,\partial_r)$ will be bounded by
\begin{align}
\text{deg}_{\partial_s,\partial_r}P^{(a)}_{\lambda,\eta}(\partial_s,\partial_r) \leq |\lambda|+|\eta|+\text{len}(\lambda) + \text{len}(\eta). 
\end{align}
Combining this estimate with lemma \ref{genus extract lemma}, we see that the index $m$ of $\phi_m$ or $\psi_m$ in (\ref{crazy term 302}) will be at most
\begin{align*}
K+K'+L-1+\sum m_i +m'_i
\end{align*}
Since $\phi_m$ and $\psi_m$ have denominator exponent $2m-1$, and are substituted into the string equations according to
\begin{align*}
\phi_m =& \partial_{s}^m [h^0] V'(h+s+r h^{-1}), \\
\psi_m =& \partial_{s}^{m-1}\partial_{\sqrt{r}} [h^0] V'(h+s+ r h^{-1}),
\end{align*}
and by our inductive hypothesis, the denominator exponent of forcing terms (\ref{crazy term 302}) will be at most
\begin{align*}
& 2\left(K+K'+L-1+\sum m_i +m'_i\right)-1 + \sum\left( 8g_i-3 + 8g'_i-3\right). 
\end{align*}
This is clearly maximized by terms of either the form $(u_{2h} \text{ or }z_h)(u')^p(z')^{-p+\sum m_i}$, for which the bound is $8h-4+4\sum m_i$; or terms of the form $(u')^p(z')^{-p+\sum m_i}$ for which the bound is $8h-3$.  By the constraint (\ref{constraint 254}), this is maximized in the $8h-3$ case.  One can derive the same bound on the denominator exponents of forcing terms in the equation $0=[N^{-2g}]V'(L)_{n,n-1}$.  Taking linear combinations of equations to isolate $z_g$ and $u_{2g}$ does not increase the denominator exponent.  
\end{proof}

\subsection{Formula for $F^{(2)}$}
Following the procedure used in section \ref{string e1 sect}, one can compute $F^{(2)}$; however it is not feasible to perform this calculation by hand.  Using a computer we have computed $z_2$ by solving the string equations.  We then computed $F^{(2)}$ using the asymptotic Hirota-Szego relation
\begin{align}
F^{(2)}=& \partial_{x}^{-2}\left(\frac{z_2}{z}-\frac{z_{1}^2}{z^2}\right) -\frac{1}{12} \left(\frac{z_1}{z}\right) + \frac{1}{240} \partial_{x}^2 \log \frac{z}{x}
\end{align}
We find that the double integral $\partial_{x}^{-2}$ term can be computed in closed form.  This is surprising, and we do not believe it to be a coincidence; however we are unable to prove that this will persist for higher values of $g$.  We find that
\begin{align}
F^{(2)}=& \frac{1}{240 x^2}+ \frac{1}{5760 z^2 \left(z \left(u'\right)^2-\left(z'\right)^2\right)^4}\bigg( \nonumber \\
&-24 \left(z'\right)^{10}+96 z \left(u'\right)^2 \left(z'\right)^8+24 z z'' \left(z'\right)^8-8 z^2 z^{(3)} \left(z'\right)^7 -144 z^2 \left(u'\right)^4 \left(z'\right)^6 \nonumber \\
& -84 z^3
   \left(u''\right)^2 \left(z'\right)^6+6 z^2 \left(z''\right)^2 \left(z'\right)^6
-96 z^2 \left(u'\right)^2 z'' \left(z'\right)^6
-120 z^3 u' u^{(3)} \left(z'\right)^6\nonumber \\
& +20 z^3 z^{(4)}\left(z'\right)^6
-384 z^3 \left(u'\right)^3 u'' \left(z'\right)^5
+384 z^3 u' u'' z'' \left(z'\right)^5 \nonumber \\
& -84 z^4 u'' u^{(3)} \left(z'\right)^5
+172 z^3 \left(u'\right)^2 z^{(3)}\left(z'\right)^5
-84 z^3 z'' z^{(3)} \left(z'\right)^5 \nonumber \\
&-40 z^4 u' u^{(4)} \left(z'\right)^5
+15 z^3 \left(u'\right)^6 \left(z'\right)^4
+64 z^3 \left(z''\right)^3 \left(z'\right)^4 \nonumber \\
& -638 z^4 \left(u'\right)^2 \left(u''\right)^2 \left(z'\right)^4
-340 z^3 \left(u'\right)^2 \left(z''\right)^2 \left(z'\right)^4
+451 z^3 \left(u'\right)^4 z'' \left(z'\right)^4 \nonumber \\
&+192 z^4\left(u''\right)^2 z'' \left(z'\right)^4
+48 z^4 \left(u'\right)^3 u^{(3)} \left(z'\right)^4
+252 z^4 u' z'' u^{(3)} \left(z'\right)^4\nonumber \\
&+252 z^4 u' u'' z^{(3)} \left(z'\right)^4
-20 z^4\left(u'\right)^2 z^{(4)} \left(z'\right)^4
-256 z^5 u' \left(u''\right)^3 \left(z'\right)^3\nonumber \\
& -768 z^4 u' u'' \left(z''\right)^2 \left(z'\right)^3
+1152 z^4 \left(u'\right)^3 u'' z''\left(z'\right)^3
-168 z^5 \left(u'\right)^2 u'' u^{(3)} \left(z'\right)^3 \nonumber \\
&-152 z^4 \left(u'\right)^4 z^{(3)} \left(z'\right)^3
-168 z^4 \left(u'\right)^2 z'' z^{(3)} \left(z'\right)^3
+80 z^5 \left(u'\right)^3 u^{(4)} \left(z'\right)^3\nonumber \\
& -7 z^4 \left(u'\right)^8 \left(z'\right)^2
+384 z^4 \left(u'\right)^2 \left(z''\right)^3 \left(z'\right)^2
-68 z^5 \left(u'\right)^4 \left(u''\right)^2 \left(z'\right)^2\nonumber \\
& -430 z^4 \left(u'\right)^4 \left(z''\right)^2 \left(z'\right)^2
-2 z^4 \left(u'\right)^6 z'' \left(z'\right)^2
+1152 z^5 \left(u'\right)^2 \left(u''\right)^2 z'' \left(z'\right)^2\nonumber \\
& +96 z^5 \left(u'\right)^5 u^{(3)} \left(z'\right)^2
-168 z^5 \left(u'\right)^3 z'' u^{(3)} \left(z'\right)^2
-168 z^5 \left(u'\right)^3 u'' z^{(3)} \left(z'\right)^2 \nonumber \\
& -20 z^5 \left(u'\right)^4 z^{(4)} \left(z'\right)^2
-256 z^6 \left(u'\right)^3 \left(u''\right)^3 z'
-768 z^5 \left(u'\right)^3 u'' \left(z''\right)^2 z'\nonumber \\
&
+252 z^6 \left(u'\right)^4 u'' u^{(3)} z'
-12 z^5\left(u'\right)^6 z^{(3)} z'
+252 z^5 \left(u'\right)^4 z'' z^{(3)} z'\nonumber \\
& 
-40 z^6 \left(u'\right)^5 u^{(4)} z'
+64 z^5 \left(u'\right)^4 \left(z''\right)^3
+22 z^6 \left(u'\right)^6 \left(u''\right)^2
-4 z^5 \left(u'\right)^6 \left(z''\right)^2\nonumber \\
&
+7 z^5 \left(u'\right)^8 z''
+192 z^6 \left(u'\right)^4 \left(u''\right)^2 z''
-24 z^6 \left(u'\right)^7 u^{(3)}
-84 z^6\left(u'\right)^5 z'' u^{(3)}\nonumber \\
&
-84 z^6 \left(u'\right)^5 u'' z^{(3)}
+20 z^6 \left(u'\right)^6 z^{(4)}\bigg).
\end{align}

\appendix

\section{Identities for Motzkin paths} \label{MotzkinIdentSect}
In this section we prove some identities that were used above.  The proofs treat $r$ and $s$ like formal parameters; that is, the lemma statements and proofs are equally valid with $r$ and $s$ replaced by $u$ and $z$, respectively.
\begin{lemma}[Path reflecting, changing between monic and normalized OP's]\label{path reflect lemma} Let $G$ be a smooth function.
\begin{align} 
[h^p] G(\mpmon) =& r^{p/2} [h^p] G(\mpsym) \\
=& r^{p/2} [h^{-p}] G(\mpsym) \\
=& r^{p} [h^{-p}] G(\mpmon)
\end{align}
\end{lemma}
\begin{proof} An easy direct calculation.
\end{proof}

\begin{lemma}[Path raising/lowering identity]  \label{pathLoweringLemma} Let $G$ be a smooth function and $p$ be an integer.  Then we have the following identities, which correspond to lowering and raising the endpoint of a Motzkin path.
\begin{align}
\partial_s [h^{p}] G(\mpsym) =& r^{p/2} \partial_{r} r^{1/2-p/2} [h^{p-1}]G(\mpsym)  \label{path lowering} \\
\partial_{s}[h^p]G(\mpsym) =& r^{-p/2} \partial_{r} r^{p/2+1/2} [h^{p+1}]G(\mpsym) \label{path raising}
\end{align}
Iterating the above identities allows the right endpoint of a Motzkin path to be shifted to zero by
\begin{align}
\partial_{s}^{|p|} [h^{p}] G(\mpsym) =& r^{|p|/2} \partial_{r}^{|p|} [h^0]G(\mpsym) \label{path zeroing}
\end{align}
\end{lemma}
\begin{proof}
First we prove (\ref{path lowering}) by direct calculation.
\newcommand{\Rsym}{\sqrt{r} h + s + \sqrt{r} h^{-1}}
\newcommand{\Rsk}{ h + s + r h^{-1}}
\begin{align}
[h^{p-1}] r^{p/2} \partial_{r} B^{1-p} G(\mpsym ) 
=& [h^{p-1}] r^{p/2} \partial_{r} G(\Rsk) \nonumber \\
=& [h^{p-1}] r^{p/2}  h^{-1}G'(\Rsk) \nonumber \\
=& [h^{p}] r^{p/2} G'(\Rsk) \nonumber \\
=& [h^p] G'(\Rsym) \nonumber \\
=& \partial_s [h^p] G(\Rsym) \nonumber
\end{align} 
The raising identity (\ref{path raising}) now follows from the lowering identity (\ref{path lowering}) by applying $[h^r]G(\mpsym) =[h^{-r}]G(\mpsym)$ to both sides, then substituting $p\rightarrow -p$.   
\end{proof}

\begin{lemma}[Derivative swapping] \label{pathDerivSwapLemma}
Let $G$ be a smooth function, then
\begin{align}
r \partial_{r}^2 [h^0](\mpsym) =& \left( \partial_{s}^2 - \partial_{r} \right) [h^0] ( \mpsym)
\end{align}
\end{lemma}
\begin{proof}
\begin{align}
\partial_{\sqrt{r}}^2 [h^0] G(\mpsym)
=& [h^0](h^2 + 2 + h^{-2}) G''(\mpsym) \nonumber \\
=& \partial_{s}^2 \left([h^{-2}]+ 2[h^0] + [h^2]\right) G(\mpsym)\nonumber \\
=& \left( r \partial_{r}^2 +2 \partial_{s}^2 + r \partial_{r}^2 \right) [h^0] G(\mpsym)\nonumber
\end{align}
To get the last line above, we used the path raising/lowering identity (\ref{path zeroing}).
The lemma now follows by using $\partial_{\sqrt{r}}^2 = 4 r \partial_{r}^2 + 2\partial_{r}$ on the left hand.
\end{proof}

\begin{lemma}[Integration by parts]
\begin{align}
\sqrt{r}\partial_s [h^p](h-h^{-1})(\sqrt{r} h + s +\sqrt{r} h^{-1})^{q} =& p [h^p](\sqrt{r} h + s +\sqrt{r} h^{-1})^{q}
\end{align}
\end{lemma}
\begin{proof}
Express the series coefficient as a contour integral about $\infty$ and integrate by parts.
\end{proof}

\section{Table of string polynomials}
\begin{align} \nonumber
\begin{array}{cc|cc}
\lambda & \eta & P^{(a)}_{\lambda,\eta} & P^{(b)}_{\lambda,\eta}  \\[.75ex] \hline
& & & \\[-2ex]
\phi & \phi & 1 & 0\\[1.5ex]
1 & \phi & 0 & -\frac{1}{2}r \partial_{r} \\[1.5ex]
\phi & 1 & \frac{1}{2} \partial_{r}& 0 \\[1.5ex]
2 & \phi & \frac{1}{6} r \partial_s \partial_{r}& \frac{1}{6} r \partial_{s}^2 +\frac{1}{12}r \partial_{r} \\[1.5ex]
1+1 &\phi & \frac{1}{12}r \partial_{s}^2 \partial_{r}
& \frac{1}{12} r \partial_{s}^3 + \frac{1}{12}r \partial_s \partial_{r} \\[1.5ex]
1 & 1 & \frac{1}{6}\partial_{s}^3 
& \frac{1}{6}r \partial_{s}^2 \partial_{r} \\[1.5ex]
\phi & 2 & \frac{1}{6} \partial_{s}^2 + \frac{1}{12} \partial_{r} & \frac{1}{6}r \partial_s \partial_{r} \\[1.5ex]
\phi & 1+1 & \frac{1}{12}r^{-1} \partial_{s}^2 -\frac{1}{12} r^{-1} \partial_{r} + \frac{1}{12}  \partial_{s}^2 \partial_{r} & \frac{1}{12}\partial_{s}^3 -\frac{1}{12} \partial_s \partial_{r} \\[1.5ex]
3 & \phi & 0 & -\frac{1}{12}r \partial_{s}^2 \\[1.5ex]
2+1 & \phi & 0 & -\frac{1}{6}r \partial_{s}^3 \\[1.5ex]
1+1+1 & \phi & 0 & -\frac{1}{24} r \partial_{s}^4 \\[1.5ex]
2 & 1 & \frac{1}{12} \partial_{s}^3 & -\frac{1}{12}r \partial_{s}^2 \partial_{r} \\[1.5ex]
1+1 & 1 & \frac{1}{24} \partial_{s}^4 & -\frac{1}{12} r \partial_{s}^3 \partial_{r} \\[1.5ex]
1 & 2 & \frac{1}{12} \partial_{s}^3 & -\frac{1}{12} r \partial_{s}^2 \partial_{r} \\[1.5ex]
1 & 1+1 & \frac{1}{12} \partial_{s}^3 \partial_{r} & -\frac{1}{24} \partial_{s}^4 + \frac{1}{24} \partial_{s}^2 \partial_{r} \\[1.5ex]
\phi & 3 & \frac{1}{12} \partial_{s}^2 & 0 \\[1.5ex]
\phi & 2+1 & \frac{1}{6} \partial_{s}^2 \partial_{r}  & 0 \\[1.5ex]
\phi & 1+1+1 & \frac{1}{24}r^{-1} \partial_{s}^4 - \frac{1}{24} r^{-1} \partial_{s}^2 \partial_{r} & 0
\end{array}
\end{align}




\end{document}